\documentclass[runningheads]{llncs}
\usepackage{dsfont}
\usepackage{graphicx}
\usepackage{amsmath}
\usepackage{amssymb}
\usepackage[normalem]{ulem}
\usepackage[caption=false]{subfig}
\usepackage[noend]{algorithm2e}
\usepackage{relsize}
\usepackage{booktabs}
\usepackage{multirow}

\usepackage{color}

\usepackage{tikz}
\usetikzlibrary{automata, positioning, chains,arrows,arrows.meta, shapes.geometric, calc, shadows}
\usepackage{pgfplots}
\pgfplotsset{compat=1.18}
\usetikzlibrary{pgfplots.statistics} 

\usepackage{pgfplotstable}

\usepackage[unicode=true]{hyperref}

\newcommand{\A}{\mathcal{A}}
\renewcommand{\L}{\mathcal{L}}

\definecolor{darkblue}{rgb}{0,0,0.65}
\labelformat{algocf}{Algorithm #1}

\tikzset{
->, % makes the edges directed
>=stealth, % makes the arrow heads bold
node distance=3cm, % specifies the minimum distance between two nodes. Change if necessary.
every state/.style={thick, fill=gray!10}, % sets the properties for each ’state’ node
initial text=$ $, % sets the text that appears on the start arrow
every axis/.style={-},
}
\begin{document}
\title{ %
\texorpdfstring{$\omega$}{Omega}-regular Expression Synthesis from Transition-Based B\"{u}chi Automata
}

\author{Charles Pert \and
        Dalal Alrajeh\and
        Alessandra Russo}
\authorrunning{Pert el al.}

\institute{
Imperial College London, London, UK,\\
\email{\{charles.pert, dalal.alrajeh, a.russo\}@imperial.ac.uk}}%
\maketitle              
\begin{abstract}
A popular method for modelling reactive systems is to use $\omega$-regular languages. These languages can be represented as nondeterministic B\"{u}chi automata (NBAs) or $\omega$-regular expressions. Existing methods synthesise expressions from state-based NBAs. Synthesis from transition-based NBAs is traditionally done by transforming transition-based NBAs into state-based NBAs. This transformation, however, can increase the complexity of the synthesised expressions. This paper proposes a novel method for directly synthesising $\omega$-regular expressions from transition-based NBAs. We prove that the method is sound and complete. Our empirical results show that the $\omega$-regular expressions synthesised from transition-based NBAs are more compact than those synthesised from state-based NBAs. This is particularly the case for NBAs computed from obligation, reactivity, safety and recurrence-type LTL formulas, reporting in the latter case an average reduction of over $50\%$. We also show that our method successfully synthesises $\omega$-regular expressions from more LTL formulas when using a transition-based instead of a state-based NBA.
\end{abstract}
\section{Introduction}
\label{sec:introduction}

Behaviours of reactive systems are characterised by infinite-length execution traces, which are mainly modelled by $\omega$-regular languages~\cite{ANGLUIN201657}. Nondeterministic B\"{u}chi automata (NBAs)~\cite{Büchi1990} and $\omega$-regular expressions are compact representations for these behaviours. NBAs can be \emph{state-based}, featuring accepting states, or \emph{transition-based}, featuring accepting transitions. An NBA accepts an infinite-length execution trace (i.e., a \emph{word}) when it traverses an accepting transition or state, respectively, infinitely often. Transition-based NBAs are at least as compact as state-based NBAs and their use often leads to more natural algorithms~\cite{duret.17.hdr}. This has been observed numerous times~\cite{Giannakopoulou2002FromST,KURSHAN198759,10.1007/3-540-58338-6_97,10.1007/978-3-662-44522-8_41} but not studied specifically. This property makes transition-based NBAs a preferable candidate for synthesising $\omega$-regular expressions.

Existing methods~\cite{Baier2008,timelines} synthesise $\omega$-regular expressions from state-based NBAs. Synthesising expressions from transition-based NBAs, thus far, would require transforming the NBAs into state-based NBAs (e.g., using~\cite{pin:hal-00112831}), doubling the number of states in the worst case. We hypothesise that synthesising $\omega$-regular expressions directly from transition-based NBAs can lead to more compact expressions. 

In this paper, we propose a method for the \emph{direct} synthesis of $\omega$-regular expressions from transition-based NBAs. For a given transition-based NBA, the method considers states with at least an outgoing accepting transition, as \emph{accepting} states. It decomposes the NBA into triplets of nondeterministic finite automata (NFAs), one for each $\langle$initial-state, accepting-state$\rangle$ pair. 
For each triplet, it synthesises regular expressions from its NFAs and recomposes them into an $\omega$-regular expression which describes the $\omega$-words accepted by the associated $\langle$initial-state, accepting-state$\rangle$. The union of these $\omega$-regular expressions give the full $\omega$-regular expression \emph{synthesised} from the original transition-based NBA. We prove the correctness of our proposed method and present an algorithm that implements it. We discuss the algorithm's time complexity and the descriptional complexity~\cite{surveyonautomata+regex} of the synthesised $\omega$-regular expressions.

To substantiate our hypothesis, we empirically evaluate the benefits of our method by answering the following question: 
\emph{does synthesising $\omega$-regular expressions directly from transition-based NBAs yield more compact expressions than those synthesised from state-based NBAs}". Our experiments compare expressions synthesised from transition-based and state-based NBAs specified by linear temporal logic (LTL)~\cite{origins_ltl} formulas. We use three metrics to measure compactness: \emph{reverse Polish notation}~\cite{surveyonautomata+regex}, \emph{timeline length} and \emph{star height}~\cite{timelines}. We use \emph{reverse Polish notation}, the number of nodes in the $\omega$-regular expression's syntax tree, as a proxy for the size of the $\omega$-regular expression, the total number of symbols (including operators) in the expression~\cite{surveyonautomata+regex}. We consider two datasets; the first includes LTL formulas collected from case studies and industrially used tools~\cite{timelines}. The second includes LTL formulas representing various patterns~\cite{dwyer_patterns} used to aid the development process of software systems. Starting from LTL formulas enables us to use Spot~\cite{duret.22.cav}, a tool optimised for computing compact state-based and transition-based NBAs.

The results of our experiments show that synthesising $\omega$-regular expressions directly from transition-based NBAs preserves their compactness. Without any simplification of the synthesised $\omega$-regular expressions, we found that the average reduction in \emph{reverse Polish notation} was $13.0\%$ and $11.3\%$. By dividing the dataset of LTL formulas into their types using the temporal hierarchy~\cite{mp_hierarchy}, we determined that recurrence, obligation and reactivity-type formulas tend to yield significantly smaller expressions from transition-based than state-based NBAs.

The paper is structured as follows. Section \ref{sec:preliminaries} introduces preliminaries needed for later sections. Section \ref{sec:theory}  presents our proposed method, using a detailed example of how it works for a transition-based NBA and describes an algorithm for it. We prove its correctness in Section \ref{sec:theoreticalguarantees}. In Section \ref{sec:experiment}, we introduce the two datasets used for our experiments and present our empirical results. 
In Section \ref{sec:discussion}, we discuss our findings from the experiments.
Section \ref{sec:conclusion} concludes the paper with an outline of future work.

\section{Preliminaries}
\label{sec:preliminaries}

We present the main concepts, terminologies and notations used throughout the paper. These are mainly adapted from~\cite{Baier2008},~\cite{introtolanguagesautomata} and~\cite{pin:hal-00112831}.

%We use the terminology as in~\cite{introtolanguagesautomata} and~\cite{pin:hal-00112831}. 
Let $\Sigma$ be a finite set of \emph{symbols}, called an \emph{alphabet}. A finite composition of symbols from $\Sigma$ is called a \emph{word}. Words are elements of $\Sigma^\ast$, the set of all finite words. A \emph{language} is a subset of $\Sigma^\ast$. In this paper, we consider only the following operations for composing languages. Let $A_1 \subseteq \Sigma^\ast$ and $A_2 \subseteq \Sigma^\ast$. We have: 
\begin{itemize}
    \item the \emph{union} operator, denoted as ($+$), such that $A_1+A_2 = \{u \in \Sigma^\ast |\, u \in A_1 \;\mbox{or}\; u \in A_2\}$;\footnote{Note that we use $+$ to denote the union operation, whereas some literature uses $\cup$.}
    \item the \emph{concatenation} operator, denoted as ($\cdot$), such that $A_1 \cdot A_2 = \{uv \in \Sigma^\ast |\, u \in A_1 \;\mbox{and}\; v\in A_2\}$; 
    \item the \emph{Kleene star} operator, denoted as the superscript ($^\ast$), such that $A_1^\ast$ is the language given by the union of finite concatenations of $A_1$. We concisely express this as: $A_1^\ast = \sum_{n \geq 0} A_1^n$, where $A_1^0 = \{\epsilon\}$ is the singleton set containing the empty word $\epsilon$, and $A_1^{i+1} = A_1^i \cdot A_1$ for $i \geq 0$;
    \item the \emph{wreath product} denoted as the superscript ($^\omega$), such that  $A_1^\omega$ represents an infinite concatenation of words from $A_1$.
\end{itemize} %
For example, $\{a \cdot (b)^\ast+ b \cdot a\}$ is the language made up of symbols from $\Sigma=\{a,b\}$, consisting of the words where $a$ is concatenated with a finite number of $b$ and $b$ is concatenated with $a$. We say two languages agree when a word is in one language if and only if it is in the other: $\{a \cdot (b)^\ast+ b \cdot a\}$ agrees with $\{a + a\cdot b\cdot (b)^\ast + b\cdot a\}$. We omit the concatenation operator and brackets in the absence of ambiguity.

Given a $\Sigma$, a \emph{regular} language is a formal language defined inductively from the empty language, $\emptyset$, and the singleton sets $\{\epsilon\}$ and $\{u\}$, for all $u \in \Sigma$, as base cases; and closed under union, concatenation and the Kleene star.

\begin{definition}[Nondeterministic Finite Automaton]
    A nondeterministic finite automaton (NFA) $\A$ is a tuple, $\A = (Q, \Sigma, \Delta, Q_0, F)$, where $Q$ is a finite set of states; $\Sigma$ is a finite alphabet; $\Delta\subseteq Q \times \Sigma \times Q$ is a transition relation; $Q_0 \subseteq Q$ is a set of initial states and $F \subseteq Q$ is a set of accepting states.
\end{definition}

Let $\A$ be an NFA and $w \in \Sigma^\ast$. A \emph{run} of $w$ in $\A$ is a sequence of states starting from an initial state in $Q_0$, followed by the states transitioned to as $w$ is read by $\A$. Given a run of $w$, the $k$-th consecutive pair of states $i,j$ in the run corresponds to the transition $(i,t,j)$ where $t$ is the $k$-th symbol in $w$. We say that the run \emph{takes} the transition $(i,t,j)$ from state $i$. A word $w$ is \emph{accepted} by $\A$ if at least one of its runs ends in an accepting state. Otherwise, we say that $w$ is \emph{rejected} by $\A$. The set of words accepted by $\A$, $L(\A)$, is called the language recognised by $\A$ and will always be regular~\cite{introtolanguagesautomata}.
% An equivalent definition of regular languages is the languages that are recognised by \emph{nondeterministic finite automata} (NFAs)~\cite{introtolanguagesautomata}.

An infinite composition of symbols from $\Sigma$ is called an \emph{$\omega$-word}. $\omega$-words are elements of $\Sigma^\omega$, the set of all infinite words. An \emph{$\omega$-language} is a subset of $\Sigma^\omega$, and \emph{$\omega$-regular languages} are those that take one of the following forms: $A^\omega$, where $A$ is a regular language and $\epsilon \notin A$; $A \cdot L_1$, where $A$ is regular and $L_1$ is $\omega$-regular and $L_1 + L_2$, where both $L_1, L_2$ are $\omega$-regular. See~\cite{Baier2008} for a more detailed introduction to $\omega$-regular languages. Throughout this paper, we shall drop the regular or $\omega$-regular modifiers when the language's type is clear.

While there are many types of $\omega$-automata, the most similar extension to an NFA is the nondeterministic B\"{u}chi automaton (NBA).

\begin{definition}[Transition-based NBA]
A transition-based NBA is a tuple, $B=(Q, \Sigma, \Delta, Q_0, Acc)$, where $Q$ is a finite set of states; $\Sigma$ is a finite alphabet; $\Delta\subseteq Q\times\Sigma\times Q$ is a transition relation; $Q_0 \subseteq Q$ is a set of initial states and $Acc \subseteq \Delta$ is the set of accepting transitions. Transitions in an NBA are called rejecting if they are not accepting. We denote with $\tilde{F} = \{q \in Q\ |\ (q,t,y) \in Acc\}$ the set of states in $Q$ that have at least one outgoing accepting transition.
\end{definition}

A run of an $\omega$-word $\sigma$ in an NBA $B$ is the same as the run of a word in an NFA, except the length of the run is infinite. $\sigma$ is \emph{accepted} by $B$ if and only if an accepting transition is traversed an infinite number of times in any of its runs. Otherwise, $\sigma$ is said to be \emph{rejected}. The language recognised by $B$ is the set of all of $B$'s accepted $\omega$-words. We denote it as $L_\omega(B)$. An $\omega$-language is $\omega$-regular~\cite{Baier2008}. Throughout the paper, we use $B$ to denote a transition-based NBA, and $\A$ to denote an NFA. A \emph{state-based} NBA is defined similarly to a transition-based NBA, with the exception that $Acc$ is replaced by $F\subseteq Q$, a set of accepting states. An $\omega$-word is accepted by a state-based NBA if and only if it visits an accepting state infinitely often for at least one of its runs.

 We adopt the standard convention~\cite{Baier2008,duret.22.cav,introtolanguagesautomata,pin:hal-00112831} for representing automata with graphs where: nodes are states, edges are transitions, states with incoming arrows that have no source state are initial states and circled states are accepting states. Accepting transitions have a circle around their label (see Fig.~\ref{fig:looping_aut} for an example of a transition-based NBA). 

\section{Method}
\label{sec:theory}

Previous work~\cite{Baier2008,timelines} have considered the synthesis of $\omega$-regular expressions from a state-based NBA by decomposing it into pairs of NFAs and synthesising regular expressions from them. Specifically, given a state-based NBA $B'=(Q, \Sigma, \Delta, Q_0, F)$, $L_\omega(B')$ agrees with the language
\begin{equation}
    \mathlarger{\sum}_{q_0\in Q_0,\, q \in F} \L_{q_0 q} \cdot (\L_{qq} \setminus \{\epsilon\})^\omega,
    \label{eq:state-basedexp}
\end{equation} % 
where $\L_{ij}$ is the set of words that have runs from state $i$ to state $j$ in $B'$ and recognised by the NFA $\A_{ij} = (Q, \Sigma, \Delta, \{i\}, \{j\})$~\cite[Lemma 4.39]{Baier2008}.

However, the language given by the above equation does not agree with those recognised by transition-based NBAs. For instance, consider the transition-based NBA $B_1$ illustrated in Fig.~\ref{fig:looping_aut} and use $\tilde{F}$ in place of $F$ in Eq.~\ref{eq:state-basedexp}, making the language $\L_{01} \cdot (\L_{11} \setminus \{\epsilon\})^\omega + \L_{02} \cdot (\L_{22} \setminus \{\epsilon\})^\omega$. This language would contain $\omega$-words that the NBA does not accept. Take state $1$ of $B_1$ as an example. $\L_{11} \setminus \{\epsilon\}$ would include $c$, making $a(c)^\omega$ an element of $\L_{01} \cdot (\L_{11} \setminus \{\epsilon\})^\omega $ that is not accepted by $B_1$. 

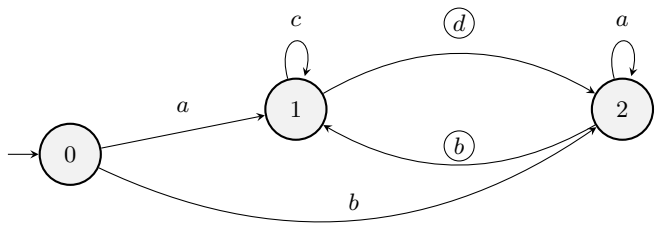
\begin{figure}[!ht]
    \centering
    \begin{tikzpicture}[every node/.style={inner sep=0pt}]
    \node[state, initial] (q0) {$0$};
    \node[state, above =0.6cm of q0, right of=q0] (q1) {$1$};
    \node[state, below =0.2cm of q1, right =3.5cm of q1] (q2) {$2$};
    \draw (q1) edge[loop above] node[yshift=0.2cm] (acc0) {$c$} (q1)
    (q2) edge[loop above] node[yshift=0.2cm] (acc1) {$a$} (q2)
    (q1) edge[bend left, above] node[draw, circle, yshift=0.2cm, minimum size=0.4cm, align=center] (acc2) {$d$} (q2)
    (q2) edge[bend left] node[draw, circle, yshift=0.25cm, minimum size=0.4cm, align=center] (acc3) {$b$} (q1)
    (q0) edge[above] node[yshift=0.25cm] {$a$} (q1)
    (q0) edge[bend right] node[yshift=0.25cm] {$b$} (q2);
    \end{tikzpicture}
    \caption{The transition-based NBA $B_1$. Edges with circled labels indicate accepting transitions.}
    \label{fig:looping_aut}
\end{figure}

We now present our method for synthesising an $\omega$-regular expression from a transition-based NBA. We use $B_1$ as a running example to illustrate our method. Our method makes use of the following three regular languages $\L_{ij,\,\text{all}}$, $\L_{ij,\,\text{rej}}$ and $\L_{ij,\,\text{acc}}$ to capture the semantics of transition-based acceptance, where $B$ is a transition-based NBA: 
\begin{itemize}
    \item $\L_{ij,\,\text{all}}$ is the set of nonempty words that have a run in $B$ from state $i$ ending upon reaching state $j$;
    \item $\L_{ij,\,\text{rej}}$ is the subset of $\L_{ij,\,\text{all}}$ given by words that have runs in $B$ that only take rejecting transitions from state $i$;
    \item $\L_{ij,\,\text{acc}}$ is the subset of $\L_{ij,\,\text{all}}$ given by words that have runs in $B$ that only take accepting transitions from state $i$.
\end{itemize} %
For example, consider $\L_{01,\,\text{all}}$ of $B_1$; the words in this language have runs from state $0$ to state $1$, visiting state $2$ finitely many times. However, it does not contain words that have runs from state $0$ to state $1$ that visit state $1$ multiple times. For instance, the word $acdab$ is not in $\L_{01,\,\text{all}}$, which can instead be described by the expression $a + b a^\ast b$.

The above three languages are sufficient to synthesise an $\omega$-regular expression that agrees with the language recognised by a transition-based NBA. Let $B=(Q, \Sigma, \Delta, Q_0, Acc)$ be a transition-based NBA and consider the $\omega$-regular language given by:
\begin{equation}
    \mathlarger{\sum}_{q_0\in Q_0, q\in \tilde{F}} \L_{q_0q,\,\text{all}} \cdot ((\L_{qq,\,\text{rej}})^\ast \cdot \L_{qq,\,\text{acc}})^\omega,
    \label{eq:transitionexp}
\end{equation}
with the terms $\L_{ij,\,x}$ (where $x \in \{\text{all},\,\text{rej},\,\text{acc}\}$) as defined above. For instance, in the case of $B_{1}$, Eq.~\ref{eq:transitionexp} corresponds to the language 
%We denote pairs of $q_0$ and $q$ states as $\langle q_0$, $q\rangle$. The $\omega$-regular language given by Eq.~\ref{eq:transitionexp} for $B_1$ is:
\begin{equation*}
    \L_{01,\,\text{all}}\cdot ((\L_{11,\,\text{rej}})^\ast \cdot \L_{11,\,\text{acc}})^\omega + \L_{02,\,\text{all}}\cdot ((\L_{22,\,\text{rej}})^\ast \cdot \L_{22,\,\text{acc}})^\omega.
\end{equation*}

We generate the regular expression corresponding to $\L_{ij,\,x}$ by synthesising it from an NFA that recognises the language. We denote these NFAs as $\A_{ij,\,x}$, $x \in \{\text{all},\,\text{rej},\,\text{acc}\}$ such that $\L_{ij,\,x}$ agrees with $L(\A_{ij,x})$. These are defined as follows: %
\begin{itemize}
    \item $\A_{ij,\,\text{all}}$ is $(Q', \Sigma, \Delta', \{i\}, \{j'\})$;

    \item $\A_{ij, \text{rej}}$ is ($Q'$, $\Sigma$, $\Delta' \setminus \{(i, t, y)\ |\ (i, t,y)\in Acc'\}$, $\{i\}$, $\{j'\}$);
    
    \item $\A_{ij, \text{acc}}$ is ($Q'$, $\Sigma$, $\Delta' \setminus \{(i, t, y)\ |\ (i, t,y) \notin Acc'\}$, $\{i\}$, $\{j'\}$);
\end{itemize} where $Q' = Q \cup \{j'\}$, $\Delta' = \{(x,t,j')\, |\, (x, t, j) \in \Delta\} \cup \{(x, t, y)\, |\, (x, t, y) \in \Delta\, \land\, y \neq j\}$ and $Acc' = \{(x,t,j')\, |\, (x, t, j) \in Acc\} \cup \{(x, t, y)\, |\, (x, t, y) \in Acc\, \land\, y \neq j\}$. These NFAs contain an additional state, $j'$, which copies state $j$ in $B$ \emph{without} any outgoing transitions. This enforces that only words with runs that end upon reaching state $j$ (in $B$) will be recognised by each NFA. Introducing the $j'$ state is necessary as otherwise, the NFAs would capture the empty word when $i=j$.

Going back to our running example, the NFAs associated with the pair of states $\langle$0, 1$\rangle$ in $B_1$ is illustrated in Fig.~\ref{fig:decomp_looping_aut}. Note that the NFAs associated with the second pair $\langle$0, 2$\rangle$ are isomorphic to those associated with the pair $\langle$0, 1$\rangle$ up to the transition symbols.

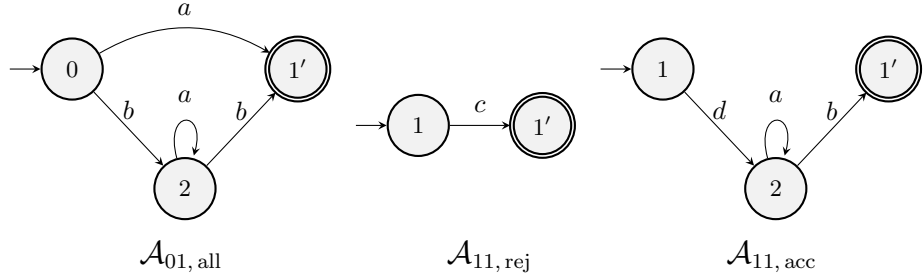
\begin{figure}[h]
    \centering
    \definecolor{brightcerulean}{rgb}{0.11, 0.67, 0.84}
    \begin{tikzpicture}[every node/.style={inner sep=0pt}]
    \node[state, initial] (q0) {$0$};
    \node[state, xshift =3cm, accepting] (q1prime) {$1'$};
    \node[state, below= 0.75cm of q0, xshift=1.5cm] (q2) {$2$};
    \node[right=1cm of q0, yshift=-2.5cm, label=center:{\large $\A_{01,\,\text{all}}$}] (label1) {};
    \draw (q0) edge[bend left] node[yshift=0.25cm]{\normalsize $a$} (q1prime);
    \draw (q0) edge[] node[yshift=0.25cm]{\normalsize $b$} (q2);
    \draw (q2) edge[loop above] node[yshift=0.25cm]{\normalsize $a$} (q2);
    \draw (q2) edge[] node[yshift=0.25cm]{\normalsize $b$} (q1prime);

    \node[state, initial, right=0.75cm of q1prime, yshift=-0.75cm] (s1) {$1$};
    \node[state, accepting, right=0.8cm of s1] (s1prime) {$1'$};
    \node[right=0.5cm of s1, yshift=-1.75cm, label=center:{\large $\A_{11,\,\text{rej}}$}] (label2) {};
    \draw (s1) edge node[yshift=0.25cm]{\normalsize $c$} (s1prime);

    \node[state, initial, right=0.75cm of s1prime, yshift=0.75cm] (r1) {$1$};
    \node[state, xshift =3cm, accepting] (r1prime) at (r1) {$1'$};
    \node[state, below= 0.75cm of r1, xshift=1.5cm] (r2) {$2$};
    \node[right=1cm of r1, yshift=-2.5cm, label=center:{\large $\A_{11,\,\text{acc}}$}] (label1) {};
    % \draw (r1) edge[bend left] node[yshift=0.25cm]{\normalsize $c$} (r1prime);
    \draw (r1) edge[] node[yshift=0.25cm]{\normalsize $d$} (r2);
    \draw (r2) edge[loop above] node[yshift=0.25cm]{\normalsize $a$} (r2);
    \draw (r2) edge[] node[yshift=0.25cm]{\normalsize $b$} (r1prime);

    \end{tikzpicture}
    \caption{The NFAs $\A_{01,\,\text{all}}$, $\A_{11,\,\text{rej}}$ and $\A_{11,\,\text{acc}}$ of $B_1$. Regular expressions that agree with these automata are $a+ba^\ast b$, $c$ and $da^\ast b$, respectively.}
    \label{fig:decomp_looping_aut}
\end{figure}

\begin{proposition}
    Let $B=(Q, \Sigma, \Delta, Q_0, Acc)$ be a transition-based NBA  and $i,j$ be two states in $Q$. Then the languages $L(\A_{ij,\,\text{all}})$, $L(\A_{ij, \text{rej}})$ and $L(\A_{ij, \text{acc}})$ agree with $\L_{ij,\,\text{all}}$, $\L_{ij,\,\text{rej}}$ and $\L_{ij,\,\text{acc}}$, respectively.
    % The language recognised by the NFA $\A_{ij, \text{rej}}$ is $\L_{ij,\,\text{rej}}$.
    % The language recognised by the NFA $\A_{ij, \text{acc}}$ is $\L_{ij,\,\text{acc}}$.
    \label{prop:nfas}
\end{proposition}

The proof of Prop.~\ref{prop:nfas} is given in Appendix.

Once the NFAs are generated from a given transition-based NBA, we can apply a method for synthesising regular expressions from these NFAs, like state elimination~\cite{introtolanguagesautomata}. We can then compose the synthesised regular expressions using Eq.~\ref{eq:transitionexp} into the overall $\omega$-regular expression that agrees with the language recognised by the original NBA. Consider again our running example and assume a method for synthesising a regular expression from an NFA. This would give the regular expressions $a+ba^\ast b$, $c$ and $da^\ast b$ for $\A_{01,\,\text{all}}$, $\A_{11,\,\text{rej}}$ and $\A_{11,\,\text{acc}}$ of $B_1$, respectively. The $\omega$-regular expression synthesised for the pair $\langle 0, 1\rangle$ is $(a+b a^\ast b)\cdot ((c)^\ast \cdot d a^\ast b)^\omega$. The overall expression synthesised for $B_1$ is $(a+b a^\ast b)\cdot ((c)^\ast \cdot d a^\ast b)^\omega + (b+a c^\ast d)\cdot ((a)^\ast \cdot b c^\ast d)^\omega$. 

We can now define an algorithm for synthesising an $\omega$-regular expression from a given transition-based NBA $B$. This is given in \ref{algo:method}. It loops over all pairs of $\langle$initial-state, accepting-state$\rangle$ in $B$ and composes synthesised regular expressions for every associated triplet of NFAs (i.e. $term$ in lines~\ref{methodline:term1} and \ref{methodline:term2}) into the final $\omega$-regular expression. The function \texttt{NFA2Regex} abstracts a method for synthesising a regular expression from an NFA. \ref{algo:method} always terminates provided $|Q|$ is finite and the method used for synthesising a regular expression from an NFA also always terminates.

\RestyleAlgo{ruled}
\SetAlgoSkip{}
\setlength{\algomargin}{2em}
\setlength{\textfloatsep}{10pt}
\begin{algorithm}[h]
    \LinesNumbered
    \SetKwFunction{NFARegex}{NFA2Regex}\SetKwFunction{Union}{Union}
    \SetKwInOut{Input}{Input}\SetKwInOut{Output}{Output}
    \Input{A nondeterministic B\"{u}chi automaton $B=(Q, \Sigma, \Delta, Q_0, Acc)$.}
    \Output{An $\omega$-regular expression describing the language recognised by $B$.}
    % \BlankLine
    $expression \longleftarrow \emptyset$\;
    $\tilde{F} \longleftarrow \{x\,|\,(x,u,y) \in Acc\}$\;
        \ForEach{$q_0 \in Q_0$ }{\label{methodline:outerloop}
            \ForEach{$q \in \tilde{F}$}{\label{methodline:innerloop}
                $term \longleftarrow \emptyset$\;
                \uIf{$q \neq q_0$}{
                    $term \longleftarrow \NFARegex(\A_{q_0q,\,\text{all}})\cdot ((\NFARegex(\A_{qq,\,\text{rej}}))^\ast \cdot \NFARegex(\A_{qq,\,\text{acc}}))^\omega$\;\label{methodline:term1}
                }\uElse{
                $term \longleftarrow ((\NFARegex(\A_{qq,\,\text{rej}}))^\ast \cdot \NFARegex(\A_{qq,\,\text{acc}}))^\omega$\;\label{methodline:term2}
                }
                $expression \longleftarrow expression + term$\;\label{methodline:compose}
            }
            }
        \Return $expression$\;
    \caption{$\omega$-regular expression synthesis for transition-based NBAs\label{algo:method}.}
\end{algorithm}

It is interesting to notice that if we treat a state-based NBA $(Q, \Sigma, \Delta, Q_0, F)$ as its equivalent transition-based $B = (Q, \Sigma, \Delta, Q_0, Acc=\{(f, t, y)\in \Delta\ |\ f \in F\}$, then Eq.~\ref{eq:transitionexp} generalises to both acceptance types of NBAs. Specifically, Eq.~\ref{eq:transitionexp} would reduce to an isomorphic version of Eq.~\ref{eq:state-basedexp}. 

In the following section, we prove that the language given by Eq.~\ref{eq:transitionexp} agrees with $L_\omega(B)$.

\section{Theoretical Results}
\label{sec:theoreticalguarantees}

Intuitively, the language given by Eq.~\ref{eq:transitionexp} for a given transition-based NBA $B$ agrees with the language recognised by $B$ considering the fact that accepted $\omega$-words must have a run with three fundamental components:
\begin{enumerate}
    \item a finite prefix: a word with a run from an initial state $q_0$ that ends upon reaching $q$, a state with at least one outgoing transition;
    \item a finite (i.e., possibly zero) number of nonempty words with a run from state $q$ that ends upon reaching $q$ that only takes a rejecting transition from $q$; \label{step:2}
    \item a nonempty word with a run from state $q$ that ends upon reaching $q$ that only takes an accepting transition from $q$; steps \ref{step:2} and \ref{step:3} are repeated infinitely many times because an accepted $\omega$-word traverses an accepting transition infinitely many times.
    \label{step:3}
\end{enumerate}

We prove the soundness and completeness of our proposed method.

\begin{theorem}[Soundness and completeness]
    Let $B$ be a transition-based NBA. The $\omega$-regular language given by Eq.~\ref{eq:transitionexp} for $B$ agrees with $L_\omega(B)$.
\end{theorem}

\begin{proof}
    Let $\sigma\in L_\omega(B)$, i.e. it is accepted by $B$. We show that it is an element of the language given by Eq.~\ref{eq:transitionexp}. By assumption, $\sigma$ has a run starting from an initial state $q_0$ and traverses an outgoing accepting transition from a state $q$ infinitely many times. Consider the decomposition of $\sigma$ into nonempty finite words that have runs that end upon reaching $q$: $\sigma = u \cdot v_1 \cdot v_2 \dots$ where $u, v_i \in \Sigma^\ast \setminus \{\epsilon\}$. If $q \neq q_0$, $u$ is an element of $\L_{q_0 q,\,\text{all}}$ by construction, otherwise, we treat $u$ as the first $v_i$. Since $\sigma$ is accepted, there is a finite word $v_{n}$ in $\sigma$ that is the first with a run from state $q$ that ends upon reaching $q$ taking an outgoing accepting transition from $q$, i.e. $v_{n} \in L_{qq,\,\text{acc}}$. It must be the case (by construction) that each $v_{1},\dots,v_{n-1}$ are elements of $\L_{qq,\,\text{rej}}$. Therefore, the decomposition  $v_{1} \cdot (\dots) \cdot v_{n}$ is by construction an element of $(\L_{qq,\,\text{rej}})^\ast \cdot \L_{qq,\,\text{acc}}$. Because $\sigma$ is accepted there must be infinitely many of these $v_n$ words. The entire decomposition $v_1 \dots v_{n+1}\dots$ is therefore a word in $((\L_{qq,\,\text{rej}})^\ast \cdot \L_{qq,\,\text{acc}})^\omega$. Hence, by construction, $\sigma \in \L_{q_0 q,\,\text{all}} \cdot ((\L_{qq,\,\text{rej}})^\ast \cdot \L_{qq,\,\text{acc}})^\omega \subseteq \sum_{q_0\in Q_0, q\in \tilde{F}} \L_{q_0q,\,\text{all}} \cdot ((\L_{qq,\,\text{rej}})^\ast \cdot \L_{qq,\,\text{acc}})^\omega$. 
    
    Let $\sigma \in \Sigma^\omega$ be an $\omega$-word that is an element of the language given by Eq.~\ref{eq:transitionexp} for $B$. We show that $\sigma \in L_\omega(B)$. Assume that $\sigma \notin L_\omega(B)$. $\sigma$ must be an element of some $\L_{q_0q,\,\text{all}} \cdot ((\L_{qq,\,\text{rej}})^\ast \cdot \L_{qq,\,\text{acc}})^\omega$, for some $q_0$ initial state in $B$ and state $q$ with at least one outgoing accepting transition in $B$ . By assumption, $\sigma$ must not have any runs that infinitely traverse an accepting transition from $q$. However, by construction, $\sigma$ must have a run from $q_0$ that reaches $q$ and infinitely traverses an accepting transition from $q$ due to the term $\L_{qq,\,\text{acc}}$. We reach a contradiction: $\sigma$ must be accepted by $B$ and $\sigma \in L_\omega(B)$. \qed
\end{proof}

\subsubsection{Time complexity} The time complexity of synthesising a regular expression from an NFA is $\mathcal{O}(n^3)$ where $n$ is the number of states in the NFA. This result is obtained by recognising that the problem is isomorphic to the all-pairs shortest path problem~\cite{Rote1985} which takes $\mathcal{O}(n^3)$ time using the Floyd-Warshall algorithm~\cite{intro_to_algorithms}. \ref{algo:method} uses a nested loop over the initial states and the accepting states meaning the maximum number of iterations is $\mathcal{O}(|Q|^2)$ where $|Q|$ is the number of states in the NBA. Each iteration generates a triplet of NFAs and synthesises regular expressions from these NFAs. So each iteration has a time complexity of the order $\mathcal{O}(3(|Q|+1)^3)$. The overall time complexity is $\mathcal{O}(|Q|^2) \times \mathcal{O}(3(|Q|+1)^3)$ which is of the same class as $\mathcal{O}(|Q|^5)$. In practice, the time complexity is closer to $\mathcal{O}(|Q|^4)$ because NBAs usually have one initial state.

\subsubsection{Descriptional complexity} We use the size of the synthesised expressions to measure descriptional complexity. The work in~\cite{surveyonautomata+regex} proves that $\mathcal{O}(|\Sigma| 2^{\Theta(n)})$ is necessary and sufficient for a regular expression describing the language of an NFA. \ref{algo:method} inherits this complexity because we use regular expressions to construct the $\omega$-regular expression. Specifically, the worst-case descriptional complexity of the $\omega$-regular expressions synthesised by \ref{algo:method} is $\mathcal{O}(|Q|^2) \times \mathcal{O}(3 |\Sigma| 2^{\Theta(|Q|+1)}) = \mathcal{O}(|Q|^2 |\Sigma| 2^{\Theta(|Q|+1)})$. 

\section{Experimental Evaluation}
\label{sec:experiment}

This section empirically substantiates the practical benefits of synthesising $\omega$-regular expressions from transition-based NBAs instead of state-based NBAs. For this, we considered the dataset of LTL formulas presented in~\cite{timelines}, their respective state-based and transition-based NBAs, computed using Spot~\cite{duret.22.cav}, and evaluated the $\omega$-regular expressions synthesised from both NBAs. The evaluation aims to answer the following questions:
\begin{itemize}
    \setlength{\itemsep}{0.5em}
    \setlength{\itemindent}{0.35cm}
    \setlength{\labelwidth}{2em}
    \setlength{\labelsep}{0.5em}
    \item [\textbf{RQ1}]Do transition-based NBAs give smaller $\omega$-regular expressions compared to those from state-based NBAs?
    \item [\textbf{RQ2}]Are there specific types of LTL formulas that give more compact $\omega$-regular expressions from transition-based NBAs as opposed to state-based NBAs?
    \item [\textbf{RQ3}]Are there characteristics or patterns in LTL formulas that indicate whether transition-based NBAs will produce more compact $\omega$-regular expressions versus state-based NBAs?
    \item [\textbf{RQ4}]Does the compactness of transition-based NBAs %enable better scaling when processing more complex LTL formulas?
    enable processing more complex LTL formulas instead of state-based NBAs in the same time limit?
\end{itemize}

Throughout the experiments, we used the state elimination method~\cite{introtolanguagesautomata} for synthesising regular expressions from NFAs. Briefly, the method iteratively eliminates the states between an initial and accepting state resulting in a two-state automaton with one transition labelled with an equivalent regular expression to the original NFA. We used the algorithm and implementation from~\cite{timelines} to synthesise expressions from state-based NBAs. Their method synthesises $\omega$-regular expressions from state-based NBAs computed using Spot from an LTL formula. Our method used transition-based NBAs computed using Spot. We used a machine with 32GB RAM, an Intel Core i7-1260P processor and version 2.11.6 of Spot. We allocated a 120-second limit per formula for Spot to compute the NBA and its syntax tree (i.e., the tree representing the $\omega$-regular expression) --- this time limit included simplification when used. A further 120-second limit was imposed when determining each metric from the expression's syntax tree. Furthermore, we used two approaches once the expression was computed. In the first, we computed the metrics with \emph{no simplification} of the syntax tree. In the second, the syntax tree is simplified using a heuristic of 8 common identities for ($\omega$-regular and regular) expressions taken from~\cite{timelines}. These include $x+xy^\ast \Rightarrow xy^\ast$ and $xyy^\omega \Rightarrow xy^\omega$, see~\cite{timelines} for the remainder.

\subsubsection{Metrics} We compared the expressions computed from NBAs with transition-based and state-based acceptance using three metrics: \emph{reverse Polish notation} (\emph{rpn}), \emph{timeline length} (\emph{tllen}) and \emph{star height} (\emph{h}). The \emph{rpn} of an expression is the total number of nodes in its syntax tree (i.e., its size without parentheses). %
% We did not use \emph{size} because processing large syntax trees resulted in the time out of many samples. %
%\toremove{We do not use \emph{size} (the total number of symbols in the fully bracketed ``in-order" expression) as evaluating the nodes in the syntax tree was much faster practically than evaluating the length of the in-order string representing large syntax trees. \emph{rpn} and \emph{size} are similar metrics quantifying the absolute size of the expression but \emph{rpn} will not capture changes in the nesting of brackets.}
The \emph{timeline length}~\cite{timelines} is the extension of the length of a regular expression (i.e., the number of symbols in the longest non-repeating path through the expression~\cite{complexity_re}). The \emph{star height} is the maximum depth of nested Kleene stars in the expression.

To ensure a fair comparison between our proposed method and the approach presented in~\cite{timelines}, we evaluated both methods using the same metrics employed in their study. This work used \emph{timeline length} and \emph{star height} to quantify the size of the proposed graphical display of the traces accepted by the LTL formula. These metrics were not ideal for our purposes as we targeted the changes to the entire expression. We mainly focused on \emph{reverse Polish notation} because it is indicative of the \emph{size} of an expression~\cite{surveyonautomata+regex} while remaining practical to compute. 

\subsubsection{Datasets}

We used the dataset provided in~\cite{timelines} to evaluate our method. The dataset comprises LTL formulas collected from~\cite{anza_benchmarks2,anza_benchmarks1,alaska_benchmarks,aac_benchmarks,acacia_benchmarks}. We include a further $98$ formulas (for a total of $185$) not evaluated in~\cite{timelines}. We also compared the synthesised expressions for common patterns for LTL formulas, as outlined by Dwyer~\cite{dwyer_patterns}, and their complements. These patterns serve as key descriptors of typical properties in LTL utilised in industrial contexts, see~\cite{patterns_servicebased} and~\cite{patterns_robots} for examples. We refer to the first dataset as the \emph{timelines} dataset and the second as the \emph{patterns} dataset. 

\subsection{Results}

In this section, we present the results of our experiments for answering the questions above. We only show the results for experiments where the expression's syntax tree is not simplified. We also present additional results that use simplification in the Appendix and we found that they are consistent with the data presented here.
% Additional results may be found in the Appendix.

\subsubsection{RQ1} Table~\ref{tab:timelines} summarises the results for the timelines dataset with and without simplification. On average, transition-based NBAs produced expressions (without simplification) $13.0\%$ smaller in \emph{rpn} than their state-based equivalent. Table~\ref{tab:timelines} also contains the number of formulas where the metric improved ($\downarrow$), stayed the same ($=$) or worsened ($\uparrow$) when using a transition-based NBA instead of a state-based NBA. We do not include formulas that timed out.
% Specifically, 73 (from 115 that did not time out) formulas gave smaller expressions from their transition-based NBAs. % We include the equivalent table for the patterns dataset in the appendix.

\begin{table}[!h]
\centering
\caption{ %
%On average, the expressions synthesised from transition-based NBAs were smaller than those from state-based NBAs. 
This table presents the results for the timelines dataset. We report the mean and standard deviations for the percentage decrease in the metric when using transition-based NBAs instead of state-based NBAs to synthesise $\omega$-regular expressions.}
\label{tab:timelines}
\resizebox{\textwidth}{!}{\begin{tabular}{@{}clccclccclccc@{}}
\toprule
Simplified & \multicolumn{4}{c}{\emph{rpn}} & \multicolumn{4}{c}{\emph{tllen}} & \multicolumn{4}{c}{\emph{h}}  \\ \midrule
- &
  \multicolumn{1}{c}{$\ $Decrease (\%)$\ $} &
  $\ \downarrow\ $ &
  $\ =\ $ &
  $\ \uparrow\ $ &
  \multicolumn{1}{c}{$\ $Decrease (\%)$\ $} &
  $\ \downarrow\ $ &
  $\ =\ $ &
  $\ \uparrow\ $ &
  \multicolumn{1}{c}{$\ $Decrease (\%)$\ $} &
  $\ \downarrow\ $ &
  $\ =\ $ &
  $\ \uparrow\ $ \\ \midrule
No         & $13.0 \pm 26.1$ & 73 & 40 & 2 & $11.8 \pm 27.8$  & 45  & 67 & 3 & $2.3 \pm 13.5$ & 19 & 85 & 0 \\
Yes        & $9.8 \pm 22.9$  & 48 & 60 & 2 & $9.3 \pm 2.41$   & 29  & 80 & 1 & $2.3 \pm 13.5$ & 15 & 83 & 0 \\ \bottomrule
\end{tabular}}
\end{table}

We depict scatter plots in Fig.~\ref{fig:scatter_metrics_nosimp_tls} for each metric we evaluated. Each circle in the plot represents an LTL formula from the timelines dataset and indicates the values of the metrics for the expressions (without simplification) when synthesised from transition-based and state-based NBAs. % Scatter plots for the other experiments (i.e., with simplification and for the patterns dataset) can be found in the appendix.
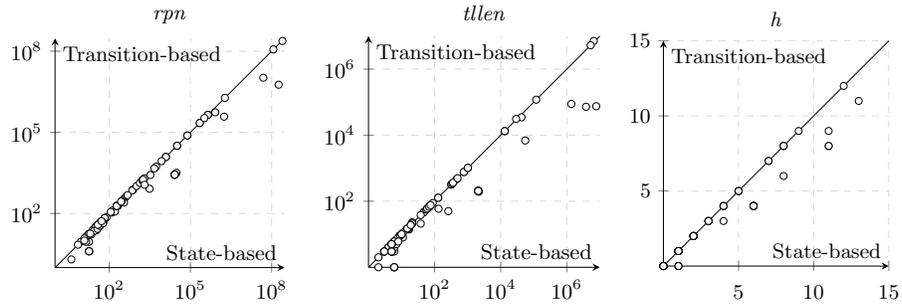
\begin{figure}[!h]
    \centering
    \captionsetup[subfigure]{labelformat=empty}
    \resizebox{\textwidth}{!}{%
    \begin{tabular}{ccc}
    \resizebox{0.33\textwidth}{!}{%
        % \subfloat[]{ %
        \begin{tikzpicture}
            \begin{axis}[
                axis x line=middle,
                axis y line=middle,
                xmin=1, xmax=3e8,
                xmode=log,
                ymode=log,
                width=5cm, height=5cm,     % size of the image
                grid = major,
                grid style={dashed, gray!30},
                title={\emph{rpn}},
                title style={yshift=-1mm},
                ylabel style={align=flush left},
                ylabel={Transition-based},
                xlabel=State-based,
                xtick = {1e2, 1e5, 1e8},
                ytick= {1e2, 1e5, 1e8},
                legend style={at={(0.5,-0.05)}, anchor=north,
                legend columns = 2, /tikz/every even column/.append style={column sep=0.25cm}},
                % every axis plot/.append style=only marks
             ] %
            \addplot[only marks,fill opacity=0.75, fill=white, mark size=1.5pt] table [x=nosimplify size, y=transition size, col sep=comma] {data/timelines_data.csv};
            % \addlegendentry{Transition-based NBA}
            %         \addplot[only marks, fill=red!40] table [x=nosimplify size
            % , y=simptrans size
            % , col sep=comma] {data/timelines_data.csv};
            % \addlegendentry{Simplified exp. from transition-based NBA}
                \addplot[post/.style={mark=none},domain=1:3e8, samples=100, black]{x};
            \end{axis}
        \end{tikzpicture}
        % } %
    } & \resizebox{0.33\textwidth}{!}{ %
     % \subfloat[]{ %
        \begin{tikzpicture}
            \begin{axis}[
                axis x line=middle,
                axis y line=middle,
                xmin=1,
                xmode=log,
                ymode=log,
                width=5cm, height=5cm,     % size of the image
                grid = major,
                grid style={dashed, gray!30},
                title={\emph{tllen}},
                title style={yshift=-1mm},
                ylabel style={align=flush left},
                ylabel={Transition-based},
                xlabel=State-based,
                legend style={at={(0.5,-0.05)}, anchor=north,
                legend columns = 2, /tikz/every even column/.append style={column sep=0.25cm}},
                % every axis plot/.append style=only marks
             ] %
            \addplot[only marks,fill opacity=0.75, fill=white, mark size=1.5pt] table [x=nosimplify length
            , y=transition length
            , col sep=comma] {data/timelines_data.csv};
            % \addlegendentry{Transition-based NBA}
            %         \addplot[only marks, fill=red!40] table [x=nosimplify size
            % , y=simptrans size
            % , col sep=comma] {data/timelines_data.csv};
            % \addlegendentry{Simplified exp. from transition-based NBA}
                \addplot[post/.style={mark=none},domain=1:1e7, samples=100, black]{x};
                    \end{axis}
        \end{tikzpicture}
        \label{subfig:scatter_timelines}
        % } %
    } & \resizebox{0.33\textwidth}{!}{ %
        \begin{tikzpicture}
            \begin{axis}[
                axis x line=middle,
                axis y line=middle,
                xmin=0,
                width=5cm, height=5cm,     % size of the image
                grid = major,
                grid style={dashed, gray!30},
                title={\emph{h}},
                title style={yshift=-1mm},
                xticklabel style={yshift=-0.1cm},
                ylabel style={align=flush left},
                ylabel={Transition-based},
                xlabel=State-based,
                legend style={at={(0.5,-0.05)}, anchor=north,
                legend columns = 2, /tikz/every even column/.append style={column sep=0.25cm}},
                % every axis plot/.append style=only marks
             ] %
            \addplot[only marks,fill opacity=0.75, fill=white, mark size=1.5pt] table [x=nosimplify starheight
            , y=transition starheight, col sep=comma] {data/timelines_data.csv};\addplot[post/.style={mark=none},domain=0:15, samples=100, black]{x};
                    \end{axis}
        \end{tikzpicture}
        } 
    \end{tabular}%
    }
    \caption{
    %Expressions synthesised from transition-based NBAs are smaller than state-based NBAs. 
    This presents a comparison between expressions synthesised from state-based and transition-based NBAs. Each circle represents an expression synthesised from  LTL formulas in the timelines dataset without simplification. Circles below $y=x$ show that the metric was smaller for the expression synthesised from the transition-based NBA compared to the one from the state-based NBA.}
    \label{fig:scatter_metrics_nosimp_tls}
\end{figure} 

\subsubsection{RQ2} We grouped each formula by type from the temporal properties hierarchy~\cite{mp_hierarchy}. We present box plots displaying the distribution of proportional reductions in \emph{rpn} when synthesising from a transition-based NBA for each type of formula. For example, an \emph{rpn} change of $0$ indicates that the expression synthesised from the transition-based NBA had the same \emph{rpn} as the one synthesised from the state-based NBA. We observed in Fig.~\ref{fig:timelines-boxplots} a $53.4\%,\, 22.3\%$ and $20.7\%$ ($47.3\%,\, 11.7\%$ and $21.0\%$) reduction in \emph{rpn}, without (with) simplification, for recurrence, obligation and reactivity formulas, respectively. We omit \emph{tllen} and \emph{h} as they are less relevant to studying the overall compactness of the synthesised expressions.
\begin{figure}[!h]
    \centering
    % \captionsetup[subfigure]{labelformat=empty}
    % \resizebox{\textwidth}{!}{%
    \subfloat[Without simplification]{\begin{tikzpicture}
	\pgfplotstableread[col sep=comma]{data/timelines_mphierarchy.csv}\csvdata
    	\begin{axis}[
    		boxplot/draw direction = y,
    		% x axis line style = {opacity=0},
    		axis x line* = bottom,
    		axis y line = left,
    		% enlarge y limits,
    		% ymajorgrids,
    		xtick = {1, 2, 3, 4, 5, 6},
            xticklabel style={
                /pgf/number format/fixed,
                /pgf/number format/precision=1,
                /pgf/number format/fixed zerofill,
                % Adjust spacing as needed
                anchor=north,
                align=right,
                font=\scriptsize,
                yshift=-0.3cm,
                xshift=-0.5cm,
                rotate=30
            },
    		xticklabels = {Reactivity, Recurrence, Obligation, Persistence, Safety, Guarantee},
    		xtick style = {draw=none}, % Hide tick line
            ylabel style = {align=center},
    		ylabel = {\emph{rpn} change\\(proportion)},
            ymin=-1, ymax=0.2,
            width=6cm, height=4.5cm,
    	]
    		\foreach \n in {0,...,5} {
    			\addplot [-, boxplot, draw=black,mark=*, fill=white, mark size=1.5pt] table[y index=\n] {\csvdata};
    		}
    	\end{axis}
    \end{tikzpicture}\label{subfig:timelines_unsimpboxplot}} %
    \subfloat[With simplification]{    \begin{tikzpicture}
    	\pgfplotstableread[col sep=comma]{data/simptimelines_mphierarchy.csv}\csvdata
    	\begin{axis}[
    		boxplot/draw direction = y,
    		% x axis line style = {opacity=0},
    		axis x line* = bottom,
    		axis y line = left,
    		% enlarge y limits,
    		% ymajorgrids,
    		xtick = {1, 2, 3, 4, 5, 6},
            xticklabel style={
                /pgf/number format/fixed,
                /pgf/number format/precision=1,
                /pgf/number format/fixed zerofill,
                % Adjust spacing as needed
                anchor=north,
                align=right,
                font=\scriptsize,
                yshift=-0.3cm,
                xshift=-0.5cm,
                rotate=30
            },
    		xticklabels = {Reactivity, Recurrence, Obligation, Persistence, Safety, Guarantee},
    		xtick style = {draw=none}, % Hide tick line
            ymin=-1, ymax=0.2,
            width=6cm, height=4.5cm,
    	]
    		\foreach \n in {0,...,5} {
    			\addplot [-, boxplot, draw=black,mark=*, fill=white, mark size=1.5pt] table[y index=\n] {\csvdata};
    		}
    	\end{axis}
    \end{tikzpicture}
    \label{subfig:timelines_simpboxplot}
    }%}
    \caption{
    %Recurrence, obligation and reactivity-type formulas tended to yield significantly smaller expressions from transition-based NBAs. 
    These plots present the proportional reductions in \emph{rpn} for the formulas in the timelines dataset grouped by their type. Circles are outliers and represent formulas more than $1.5\times$ outside the interquartile range.}
    \label{fig:timelines-boxplots}
\end{figure}
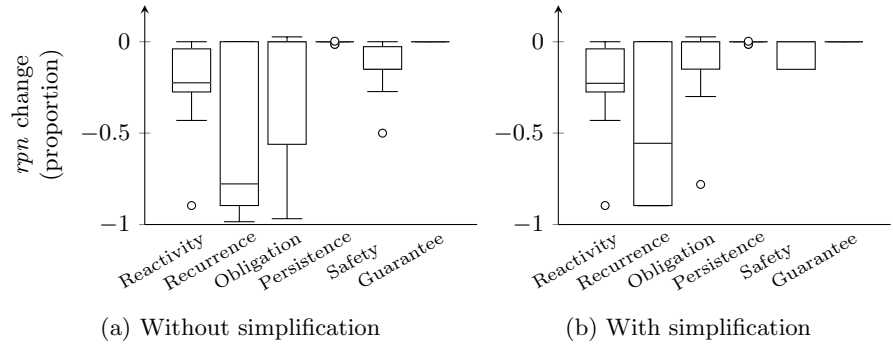

\subsubsection{RQ3} We replicated the experiments using the patterns dataset to identify the LTL patterns that give smaller $\omega$-regular expressions from transition-based NBAs. We present the results of these experiments in Fig.~\ref{fig:rpn_patterns}. These plots depict the \emph{rpn} of the expressions synthesised from both the transition-based and state-based NBA without simplification. The largest proportional reduction in \emph{rpn}, $81.7\%$ ($8307$ to $1521$), occurred for the complement of pattern 55 (with and without simplification). Due to scale, some patterns (and their complements) cannot be seen in Fig.~\ref{fig:rpn_patterns}.

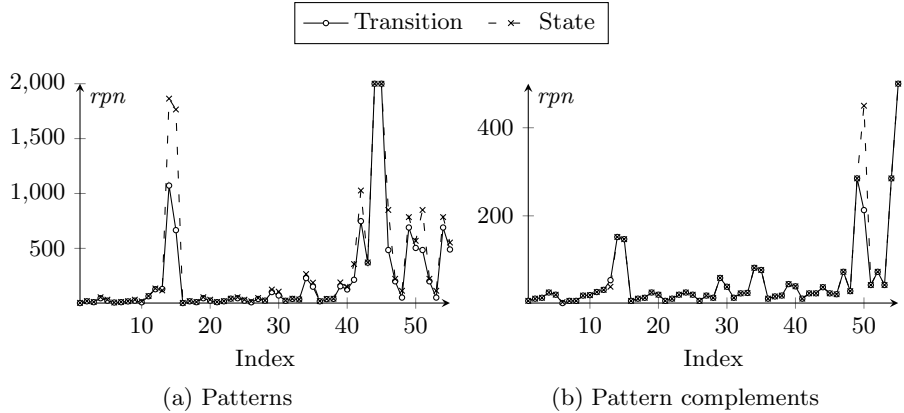
\begin{figure}[!h]
    \centering
    % \captionsetup[subfigure]{labelformat=empty}
    % \resizebox{\textwidth}{!}{%
    \begin{tikzpicture}
        \node [above] at (current bounding box.north) {\ref*{mylegend}};
    \end{tikzpicture}
    \subfloat[c][Patterns]{ %
    \begin{tikzpicture}
	\pgfplotstableread[col sep=comma]{data/patterns_rpns.csv}\csvdata
            \begin{axis}[
                axis x line=middle,
                axis y line=middle,
                xmin=1, xmax=55,
                restrict y to domain*=1:2e3,
                width=6.5cm, height=4.5cm,     % size of the image
                grid style={gray!30},
                ylabel=$\emph{rpn}$,
                xlabel={Index},
                xlabel near ticks,
                legend style={at={(0.5,-0.05)}, anchor=north,
                legend columns = 2, /tikz/every even column/.append style={column sep=0.25cm}},
                % every axis plot/.append style=only marks,
                xtick={10,20,30,40,50},
                legend to name={mylegend},
             ] %
                \addplot[-, mark size=1pt, mark=*, draw opacity=1, fill=white] table [x=Pattern
                , y=Trans rpn
                , col sep=comma] {\csvdata};
                \addlegendentry{Transition}
                \addplot[loosely dashed, -, mark=x, mark size=1.5pt, draw opacity=0.5] table [x=Pattern, y=Unsimp rpn, col sep=comma] {\csvdata};
            \addlegendentry{State}
            \end{axis}
    \end{tikzpicture}\label{subfig:patt_rpn_patterns}
    } %
    \subfloat[c][Pattern complements]{    \begin{tikzpicture}
    	\pgfplotstableread[col sep=comma]{data/patterns_rpns.csv}\csvdata
            \begin{axis}[
                axis x line=middle,
                axis y line=middle,
                xmin=1, xmax=55,
                restrict y to domain*=1:5e2,
                width=6.5cm, height=4.5cm,     % size of the image
                grid style={dashed, gray!30},
                ylabel=$\emph{rpn}$,
                xlabel=Index,
                legend style={at={(0.5,-0.05)}, anchor=north,
                legend columns = 2, /tikz/every even column/.append style={column sep=0.25cm}},
                % legend to name={mylegend},
                xlabel near ticks,
                xtick={10,20,30,40,50},
                % every axis plot/.append style=only marks
             ] %
                \addplot[-, mark size=1pt, mark=*, draw opacity=1, fill=white] table [x=Pattern
                , y=comp Trans rpn
                , col sep=comma] {\csvdata};
                % \addlegendentry{Transition-based}
                \addplot[loosely dashed, -, mark=x, mark size=1.5pt, draw opacity=0.5] table [x=Pattern, y=comp Unsimp rpn, col sep=comma] {\csvdata};
            % \addlegendentry{State-based}
            \end{axis}
    \end{tikzpicture}
    \label{subfig:comp_rpn_patterns}
    } \hfil
    % }
    % \begin{tikzpicture}
    %     \node [above] at (current bounding box.north) {\ref*{mylegend}};
    % \end{tikzpicture}
    \caption{Several patterns and complements give significantly smaller \emph{rpn} for expressions from transition-based NBAs. However, most formulas in this dataset show no or negligible improvement.}
    \label{fig:rpn_patterns}
\end{figure}
\subsubsection{RQ4} Using transition-based NBAs instead of state-based NBAs, with simplification, allowed three more LTL formulas to be evaluated for each metric from the timelines dataset. Similarly, using transition-based NBAs without simplification allowed two more formulas to be evaluated for \emph{rpn} and \emph{tllen}. % and an extra one for \emph{star height}.

\subsubsection{Summary} Our experimental results demonstrate that synthesising $\omega$-regular expressions directly from transition-based NBAs yields significantly more compact expressions than state-based NBAs. In most cases, the transition-based metrics behave as if bounded above by its associated state-based metric (see Fig.~\ref{fig:scatter_metrics_nosimp_tls}). Expressions synthesised from transition-based NBAs were, on average, $13\%$ smaller in \emph{rpn} without simplification. Recurrence, obligation and reactivity-type LTL formulas benefited the most, with reductions in \emph{rpn} of up to $98.5\%$ ($89.7\%$ with simplification). Specific LTL patterns saw significant \emph{rpn} reductions although most patterns did not see much improvement. Furthermore, transition-based NBAs enabled more LTL formulas to be evaluated within the given timeouts. These findings demonstrate the advantages of leveraging transition-based acceptance when synthesising $\omega$-regular expressions from NBAs, particularly for LTL formulas used in industrial settings.

\section{Discussion}
\label{sec:discussion}
We observed an improvement in \emph{rpn} for the synthesised $\omega$-regular expressions for formulas of the types: reactivity; recurrence; obligation and safety. This finding explains why we observed less reduction in the patterns dataset: only 21 formulas belonged to these three categories.
Recurrence-type formulas demonstrated the most improvement. Recurrence formulas were expected to give smaller $\omega$-regular expressions when using transition-based NBAs instead of state-based NBAs because they describe events occurring infinitely often~\cite{mp_hierarchy} but do not necessarily occur at every timestep. Building on our results, we predict that the LTL formulas that belong to these classes derive the most benefit when represented with transition-based NBAs. Our results suggest that applications involving these types of formulas stand to gain the most from transition-based NBA usage. 

However, our experiments indicate that persistence and guarantee-type formulas derive no benefit from being represented using a transition-based NBA: they tended to produce expressions of equal length when represented using both types of NBA. Briefly, guarantee formulas will derive no benefit because the infinite suffixes of all terms in the expression will be $\Sigma^\omega$ and persistence formulas have $\L_{qq,\,\text{rej}} = \emptyset$ for all $q \in \tilde{F}$ in their NBAs.

The proposed method using the transition-based NBAs enabled more formulas from the timelines dataset to be evaluated. More formulas timed out when the syntax tree was simplified; this was expected because the syntax trees were large and many nodes had to be traversed. However, the transition-based NBA approach with no simplification determined an extra two formulas for \emph{rpn} and \emph{tllen} but only one extra for \emph{h}. We expected that the proposed method would improve the \emph{rpn} and \emph{tllen} of the synthesised $\omega$-regular expressions but minimal improvement in \emph{h}. This agrees with our results as the average reduction of \emph{rpn} and \emph{tllen} was greater than the reduction of \emph{h}. This also explains why the proposed method only evaluated one extra formula instead of two.

During initial experimentation, we observed that $\omega$-regular expressions computed from transition-based NBAs would be larger than the state-based NBA's expression. We expected the state-based expression to behave as an upper bound for the transition-based expression. Two factors make up this issue. Firstly, Spot can produce transition-based NBAs with a non-optimal labelling of accepting transitions. This labelling increases the number of states in $\tilde{F}$ above the number of accepting states in the state-based NBA produced by Spot. Secondly, Spot can produce NBAs with different state numbering depending on the type of acceptance specified for the same LTL formula. NFA to regular expression algorithms depend on the order of states. This issue resulted in edge cases where the order of the state-based NBA was more suited to producing smaller expressions than the state-ordering of the transition-based NBA.

We mitigated the first issue by computing both automata and checking if the number of states in $\tilde{F}$ (transition-based NBA) was less than that of accepting states in the state-based NBA. If yes, we synthesised the expression from the transition-based automaton; otherwise, we used the state-based automaton (as if it were transition-based). This approach allows us to avoid the labelling of accepting transitions by Spot without impacting our results. However, we were unable to control the ordering of states, which had a minor impact on our experimental data. We believe this did not significantly affect the overall experiment and only influenced a small number of formulas. Solving this issue would require more control of the NBA generation process, specifically how Spot numbers states during the generation of an NBA.

\section{Related Work} To the best of our knowledge, no method exists for directly synthesising an $\omega$-regular expression from a transition-based B\"{u}chi automaton. Previous work, such as~\cite{Baier2008}, establishes the equivalence between $\omega$-regular languages and NBAs by constructing $\omega$-regular expressions describing the languages recognised by state-based NBAs. Similarly, other work~\cite{pin:hal-00112831} provides a method for translating a transition-based NBA into a state-based NBA. In principle, it could be possible to synthesise $\omega$-regular expressions from transition-based NBA by first translating the latter into a state-based NBA and synthesising an $\omega$-regular expression that agrees with the language recognised by this state-based NBA. However, this approach has never been investigated and has the potential to generate more complex $\omega$-regular expressions as the translation into state-based NBA can lose the compactness of
transition-based NBAs. In Section~\ref{sec:experiment} we have demonstrated and evaluated the benefit of synthesising directly from transition-based NBA. A reader might wonder whether our method of relaxing transition-based NBAs into NBAs with ``pseudo-accepting states'' might be related to methods for translating transition-based into state-based NBAs. However, the two are different, as our relaxation enables the handling of states that have both outgoing accepting and rejecting transitions while avoiding the potential blow-up of states that a full translation into an NBA with accepting states would cause. 
\section{Conclusion}
\label{sec:conclusion}

We proposed a novel method for directly synthesising  $\omega$-regular expressions from transition-based NBAs, contrary to existing work that 
%. Before this work, it was only possible to 
synthesise $\omega$-regular expressions from state-based NBAs. Our approach offers modularity in the sense that the \texttt{NFA2Regex} function (see \ref{algo:method}) enables the use of any method for synthesising regular expressions from NFAs: which allowed us to leverage existing research on synthesising regular expressions from NFAs. A comprehensive survey~\cite{surveyonautomata+regex} provides the common methods for synthesising an expression from an NFA. These methods include state elimination~\cite{stateelim1,introtolanguagesautomata} and solving characteristic equations~\cite{derivativesregex} by applying Arden's theorem~\cite{ardenstheorem,Conway1971RegularAA}.

We empirically demonstrated that using this method preserves the compactness observed in transition-based NBAs thus offering a quantifiable advantage over state-based NBAs. Our experiments confirmed this compactness property across various metrics, particularly benefiting recurrence languages. This partially alleviates the dependency on the (computationally hard~\cite{MALCHER2004375}) simplification of $\omega$-regular expressions that would otherwise be necessary when using a state-based NBA. We were also able to apply our method to %a small number of 
additional LTL formulas using transition-based NBAs that timed out when using state-based NBAs, indicating promising performance gain.% leveraging the compactness of transition-based NBAs.

Our work contributes to the understanding of the compactness observed in transition-based NBAs and our experiments demonstrate that there is potential for improved scalability when using transition-based NBAs over state-based NBAs to represent LTL formulas. Future research directions include a formal proof of the boundedness of expressions synthesised from minimal transition-based NBAs and for extending the work in~\cite{mltl_via_regexs} to LTL. We also propose that the syntax tree of an $\omega$-regular expression offers a natural means for synthesising an LTL formula from an $\omega$-regular language and could be an alternative avenue to~\cite{automata_2_ltl}.
\subsubsection{Acknowledgements}
% Hidden to preserve the anonymity of authors.
This work was supported by the UK EPSRC grants 2760033 and EP/X040518/1.

\section*{Appendix}
\label{sec:appendix}
\subsection*{Proof of Proposition~\ref{prop:nfas}}
\begin{proof}
Let $w \in \L_{ij,\,\text{all}}$. By definition, $w$ is a nonempty word with a run in $B$ from state $i$ that ends upon reaching state $j$. Consider this run of $w$ in $\A_{ij,\,\text{all}}$: it will begin in state $i$ and end upon reaching state $j'$ (because transitions into state $j$ in $B$ are transitions into state $j'$ in $\A_{ij,\,\text{all}}$). $w$ is accepted by $\A_{ij,\,\text{all}}$ because $j'$ is an accepting state. We have that $\L_{ij,\,\text{all}} \subseteq L(\A_{ij,\,\text{all}})$. 
    
Similarly, let $w \in L(\A_{ij,\,\text{all}})$. $w$ must be nonempty because $i \neq j'$. By construction, $w$ must have a run from $i$ ending upon reaching $j$ in $B$. We have that $L(\A_{ij,\,\text{all}}) \subseteq \L_{ij,\,\text{all}}$. $\L_{ij,\,\text{all}}$ agrees with $L(\A_{ij,\,\text{all}})$.
\qed
It suffices to consider the construction of the NFAs $\A_{ij,\,\text{rej}}$ and $\A_{ij,\,\text{acc}}$ to see that they agree with $\L_{ij,\,\text{rej}}$ and $\L_{ij,\,\text{acc}}$, respectively. Both NFAs recognise languages that are subsets of the language recognised by $\A_{ij,\,\text{all}}$ and we have proven that $\L_{ij,\,\text{all}}$ is the language recognised by $\A_{ij,\,\text{all}}$. The construction of $\A_{ij,\,\text{rej}}$ ($\A_{ij,\,\text{acc}}$) forces every transition taken from state $i$ to be rejecting (accepting) thus restricting the language recognised by $\A{ij,\,\text{rej}}$ ($\A{ij,\,\text{acc}}$) to $\L_{ij,\,\text{rej}}$ ($\L_{ij,\,\text{acc}}$).
\end{proof}
\subsection*{Results from additional experiments}

We provide below results from additional experiments for \textbf{RQ1} that further support the superiority of our method.  Table~\ref{tab:patterns} demonstrates results using the pattern dataset. The figures show that transition-based NBAs representing patterns tended to yield smaller expressions than the state-based NBAs.
% Table of patterns data
\begin{table}[!ht]
\centering
\caption{This table presents the results for the patterns dataset. We report the mean and standard deviations for the percentage decrease in the metric when using transition-based NBAs instead of state-based NBAs to synthesise $\omega$-regular expressions. %in the same format as Table~\ref{tab:timelines}.
}
\label{tab:patterns}
\resizebox{\textwidth}{!}{\begin{tabular}{@{}clccclccclccc@{}}
\toprule
Simplified & \multicolumn{4}{c}{\emph{rpn}}  & \multicolumn{4}{c}{\emph{tllen}} & \multicolumn{4}{c}{\emph{h}}  \\ \midrule
- &
  \multicolumn{1}{c}{$\ $Decrease (\%)$\ $} &
  $\ \downarrow\ $ &
  $\ =\ $ &
  $\ \uparrow\ $ &
  \multicolumn{1}{c}{$\ $Decrease (\%)$\ $} &
  $\ \downarrow\ $ &
  $\ =\ $ &
  $\ \uparrow\ $ &
  \multicolumn{1}{c}{$\ $Decrease (\%)$\ $} &
  $\ \downarrow\ $ &
  $\ =\ $ &
  $\ \uparrow\ $ \\ \midrule
No         & $11.3  \pm 18.4$ & 54 & 54 & 2 & $7.9 \pm 18.5$  & 27  & 81  & 2 & $2.5 \pm 10.7$ & 6 & 104 & 0 \\
Yes        & $2.6 \pm 14.2$   & 9  & 95 & 6 & $2.6 \pm 15.6$  & 9   & 97  & 4 & $2.5 \pm 10.7$ & 6 & 104 & 0 \\ \bottomrule
\end{tabular}}
\end{table}
We present a comparison of the metrics of the $\omega$-regular expressions (without simplification) synthesised from state-based and transition-based NBAs for the LTL formulas in the patterns dataset in Fig.~\ref{fig:scatter_metrics_nosimp_patterns}. We present the same plots for the two datasets with simplification in Fig.~\ref{fig:scatter_simp_metrics}.
% Plot of unsimplified patterns metrics:
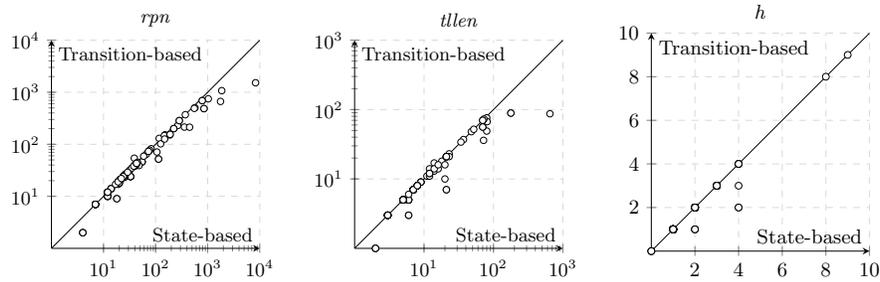
\begin{figure}[!ht]
    \centering
    \captionsetup[subfigure]{labelformat=empty}
    \resizebox{\textwidth}{!}{%
    \begin{tabular}{ccc}
    \resizebox{0.33\textwidth}{!}{ %
        \subfloat[]{ %
            \begin{tikzpicture}
                \begin{axis}[
                    axis x line=middle,
                    axis y line=middle,
                    xmin=1, xmax=1e4,
                    xmode=log,
                    ymode=log,
                    width=5cm, height=5cm,     % size of the image
                    grid = major,
                    grid style={dashed, gray!30},
                    title={\emph{rpn}},
                    title style={yshift=-1mm},
                    ylabel style={align=flush left},
                    ylabel={Transition-based},
                    xlabel=State-based,
                    legend style={at={(0.5,-0.05)}, anchor=north,
                    legend columns = 2, /tikz/every even column/.append style={column sep=0.25cm}},
                    xtick = {1e1, 1e2, 1e3, 1e4},
                    ytick = {1e1, 1e2, 1e3, 1e4},
                    % every axis plot/.append style=only marks
                 ] %
                \addplot[only marks,fill opacity=0.75, fill=white, mark size=1.5pt] table [x=nosimplify size, y=transition size, col sep=comma] {data/patterns_data.csv};
                % \addlegendentry{Transition-based NBA}
                \addplot[post/.style={mark=none},domain=1:1e4, samples=100, black]{x};
                \end{axis}
            \end{tikzpicture}
        } %
     } & %
     \resizebox{0.33\textwidth}{!}{ %
            \subfloat[]{ %
            \begin{tikzpicture}
                \begin{axis}[
                    axis x line=middle,
                    axis y line=middle,
                    xmin=1, xmax=1e3,
                    xmode=log,
                    ymode=log,
                    width=5cm, height=5cm,     % size of the image
                    grid = major,
                    grid style={dashed, gray!30},
                    title={\emph{tllen}},
                    title style={yshift=-1mm},
                    ylabel style={align=flush left},
                    ylabel={Transition-based},
                    xlabel=State-based,
                    legend style={at={(0.5,-0.05)}, anchor=north,
                    legend columns = 2, /tikz/every even column/.append style={column sep=0.25cm}},
                    xtick = {1e1, 1e2, 1e3},
                    ytick = {1e1, 1e2, 1e3},
                    % every axis plot/.append style=only marks
                 ] % 
                \addplot[only marks,fill opacity=0.75, fill=white, mark size=1.5pt] table [x=nosimplify length
                , y=transition length
                , col sep=comma] {data/patterns_data.csv};
                % \addlegendentry{Transition-based NBA}
                %         \addplot[only marks, fill=red!40] table [x=nosimplify size
                % , y=simptrans size
                % , col sep=comma] {data/timelines_data.csv};
                % \addlegendentry{Simplified exp. from transition-based NBA}
                \addplot[post/.style={mark=none},domain=1:1e3, samples=100, black]{x};
                \end{axis}
            \end{tikzpicture}
            \label{subfig:scatter_unsimp_patterns}
            } %
     } & %
     \resizebox{0.33\textwidth}{!}{ % 
            \subfloat[]{ %
            \begin{tikzpicture}
                \begin{axis}[
                    axis x line=middle,
                    axis y line=middle,
                    xmin=0,
                    width=5cm, height=5cm,     % size of the image
                    grid = major,
                    grid style={dashed, gray!30},
                    title={\emph{h}},
                    title style={yshift=-1mm},
                    ylabel style={align=flush left},
                    ylabel={Transition-based},
                    xlabel=State-based,
                    legend style={at={(0.5,-0.05)}, anchor=north,
                    legend columns = 2, /tikz/every even column/.append style={column sep=0.25cm}},
                    % every axis plot/.append style=only marks
                 ] %
                \addplot[only marks,fill opacity=0.75, fill=white, mark size=1.5pt] table [x=nosimplify starheight, y=transition starheight, col sep=comma] {data/patterns_data.csv}; %
                \addplot[post/.style={mark=none},domain=0:10, samples=100, black]{x};
                \end{axis}
            \end{tikzpicture}
            }
     }
     \end{tabular}
     }
     \caption{We observe largely the same behaviour for the formulas in the patterns dataset as in the timelines dataset. Each mark represents an LTL formula from the pattern dataset without simplification.}
    \label{fig:scatter_metrics_nosimp_patterns}
\end{figure} 

% Plot of simplified metrics:
\begin{figure}[!ht]
    \centering
    \captionsetup[subfigure]{labelformat=empty}
    \resizebox{\textwidth}{!}{%
    \begin{tabular}{ccc}
    \resizebox{0.33\textwidth}{!}{%
        \subfloat[]{ %
        \begin{tikzpicture}
            \begin{axis}[
                axis x line=middle,
                axis y line=middle,
                xmin=1, xmax=1e7,
                xmode=log,
                ymode=log,
                width=5cm, height=5cm,     % size of the image
                grid = major,
                grid style={dashed, gray!30},
                title={\emph{rpn}},
                title style={yshift=-1mm},
                ylabel style={align=flush left},
                ylabel={Transition-based},
                xlabel=State-based,
                xtick = {1e1, 1e4, 1e7},
                ytick= {1e1, 1e4, 1e7},
                legend style={at={(0.5,-0.05)}, anchor=north,
                legend columns = 2, /tikz/every even column/.append style={column sep=0.25cm}},
                % every axis plot/.append style=only marks
             ] %
            \addplot[only marks,fill opacity=0.75, fill=white, mark size=1.5pt] table [x=simplify size, y=simptrans size, col sep=comma] {data/timelines_data.csv};
            % \addlegendentry{Transition-based NBA}
            %         \addplot[only marks, fill=red!40] table [x=nosimplify size
            % , y=simptrans size
            % , col sep=comma] {data/timelines_data.csv};
            % \addlegendentry{Simplified exp. from transition-based NBA}
                \addplot[post/.style={mark=none},domain=1:1e7, samples=100, black]{x};
            \end{axis}
        \end{tikzpicture}
        } %
    }%
    & \resizebox{0.33\textwidth}{!}{ %
     \subfloat[(a) Timelines]{ %
        \begin{tikzpicture}
            \begin{axis}[
                axis x line=middle,
                axis y line=middle,
                xmin=1,
                xmode=log,
                ymode=log,
                width=5cm, height=5cm,     % size of the image
                grid = major,
                grid style={dashed, gray!30},
                title={\emph{tllen}},
                title style={yshift=-1mm},
                ylabel style={align=flush left},
                ylabel={Transition-based},
                xlabel=State-based,
                legend style={at={(0.5,-0.05)}, anchor=north,
                legend columns = 2, /tikz/every even column/.append style={column sep=0.25cm}},
                xtick={1e2, 1e4, 1e6},
                ytick={1e2, 1e4, 1e6},
                % every axis plot/.append style=only marks
             ] %
            \addplot[only marks,fill opacity=0.75, fill=white, mark size=1.5pt] table [x=simplify length
            , y=simptrans length
            , col sep=comma] {data/timelines_data.csv};
            % \addlegendentry{Transition-based NBA}
            %         \addplot[only marks, fill=red!40] table [x=nosimplify size
            % , y=simptrans size
            % , col sep=comma] {data/timelines_data.csv};
            % \addlegendentry{Simplified exp. from transition-based NBA}
                \addplot[post/.style={mark=none},domain=1:1e6, samples=100, black]{x};
                    \end{axis}
        \end{tikzpicture}
        \label{subfig:scatter_simp_timelines}
        } %
    } %
    &   \resizebox{0.33\textwidth}{!}{ %
    \subfloat[]{ %
        \begin{tikzpicture}
            \begin{axis}[
                axis x line=middle,
                axis y line=middle,
                xmin=0, xmax=10,
                width=5cm, height=5cm,     % size of the image
                grid = major,
                grid style={dashed, gray!30},
                title={\emph{h}},
                title style={yshift=-1mm},
                ylabel style={align=flush left},
                ylabel={Transition-based},
                xlabel=State-based,
                legend style={at={(0.5,-0.05)}, anchor=north,
                legend columns = 2, /tikz/every even column/.append style={column sep=0.25cm}},
                % every axis plot/.append style=only marks
             ] %
            \addplot[only marks,fill opacity=0.75, fill=white, mark size=1.5pt] table [x=simplify starheight
            , y=simptrans starheight, col sep=comma] {data/timelines_data.csv};
            % \addlegendentry{Transition-based NBA}
            %         \addplot[only marks, fill=red!40] table [x=nosimplify size
            % , y=simptrans size
            % , col sep=comma] {data/timelines_data.csv};
            % \addlegendentry{Simplified exp. from transition-based NBA}
                \addplot[post/.style={mark=none},domain=0:10, samples=100, black]{x};
                    \end{axis}
        \end{tikzpicture}
        } %
        } 
     \\ % 
     \resizebox{0.33\textwidth}{!}{ %
            \subfloat[]{ %
            \begin{tikzpicture}
                \begin{axis}[
                    axis x line=middle,
                    axis y line=middle,
                    xmin=1, xmax=1e4,
                    xmode=log,
                    ymode=log,
                    width=5cm, height=5cm,     % size of the image
                    grid = major,
                    grid style={dashed, gray!30},
                    title={\emph{rpn}},
                    title style={yshift=-1mm},
                    ylabel style={align=flush left},
                    ylabel={Transition-based},
                    xlabel=State-based,
                    legend style={at={(0.5,-0.05)}, anchor=north,
                    legend columns = 2, /tikz/every even column/.append style={column sep=0.25cm}},
                    xtick = {1e1, 1e2, 1e3, 1e4},
                    ytick = {1e1, 1e2, 1e3, 1e4},
                    % every axis plot/.append style=only marks
                 ] %
                \addplot[only marks,fill opacity=0.75, fill=white, mark size=1.5pt] table [x=simplify size, y=simptrans size, col sep=comma] {data/patterns_data.csv};
                % \addlegendentry{Transition-based NBA}
                \addplot[post/.style={mark=none},domain=1:1e4, samples=100, black]{x};
                \end{axis}
            \end{tikzpicture}
            } %
     } & %
     \resizebox{0.33\textwidth}{!}{ %
            \subfloat[(b) Patterns]{ %
            \begin{tikzpicture}
                \begin{axis}[
                    axis x line=middle,
                    axis y line=middle,
                    xmin=1, xmax=1e3,
                    xmode=log,
                    ymode=log,
                    width=5cm, height=5cm,     % size of the image
                    grid = major,
                    grid style={dashed, gray!30},
                    title={\emph{tllen}},
                    title style={yshift=-1mm},
                    ylabel style={align=flush left},
                    ylabel={Transition-based},
                    xlabel=State-based,
                    legend style={at={(0.5,-0.05)}, anchor=north,
                    legend columns = 2, /tikz/every even column/.append style={column sep=0.25cm}},
                    xtick = {1e1, 1e2, 1e3},
                    ytick = {1e1, 1e2, 1e3},
                    % every axis plot/.append style=only marks
                 ] % 
                \addplot[only marks,fill opacity=0.75, fill=white, mark size=1.5pt] table [x=simplify length
                , y=simptrans length
                , col sep=comma] {data/patterns_data.csv};
                % \addlegendentry{Transition-based NBA}
                %         \addplot[only marks, fill=red!40] table [x=nosimplify size
                % , y=simptrans size
                % , col sep=comma] {data/timelines_data.csv};
                % \addlegendentry{Simplified exp. from transition-based NBA}
                \addplot[post/.style={mark=none},domain=1:1e3, samples=100, black]{x};
                \end{axis}
            \end{tikzpicture}
            \label{subfig:scatter_simp_patterns}
            } %
     } & %
     \resizebox{0.33\textwidth}{!}{ % 
            \subfloat[]{ %
            \begin{tikzpicture}
                \begin{axis}[
                    axis x line=middle,
                    axis y line=middle,
                    xmin=0,
                    width=5cm, height=5cm,     % size of the image
                    grid = major,
                    grid style={dashed, gray!30},
                    title={\emph{h}},
                    title style={yshift=-1mm},
                    ylabel style={align=flush left},
                    ylabel={Transition-based},
                    xlabel=State-based,
                    legend style={at={(0.5,-0.05)}, anchor=north,
                    legend columns = 2, /tikz/every even column/.append style={column sep=0.25cm}},
                    % every axis plot/.append style=only marks
                 ] %
                \addplot[only marks,fill opacity=0.75, fill=white, mark size=1.5pt] table [x=simplify starheight, y=simptrans starheight, col sep=comma] {data/patterns_data.csv}; %
                \addplot[post/.style={mark=none},domain=0:10, samples=100, black]{x};
                \end{axis}
            \end{tikzpicture}
            }
     } 
    \end{tabular}%
    }
    \caption{It remains beneficial to use transition-based NBAs instead of state-based NBAs to synthesise $\omega$-regular expressions with simplification. While the average difference between the two synthesised expressions is less than the case without simplification, some LTL formulas still produce smaller $\omega$-regular expressions when synthesising from their transition-based NBA. We also note that simplification breaks the observed ``upper bound" behaviour of the metrics from the state-based NBA observed in the experiments without simplification.}
    \label{fig:scatter_simp_metrics}
\end{figure}
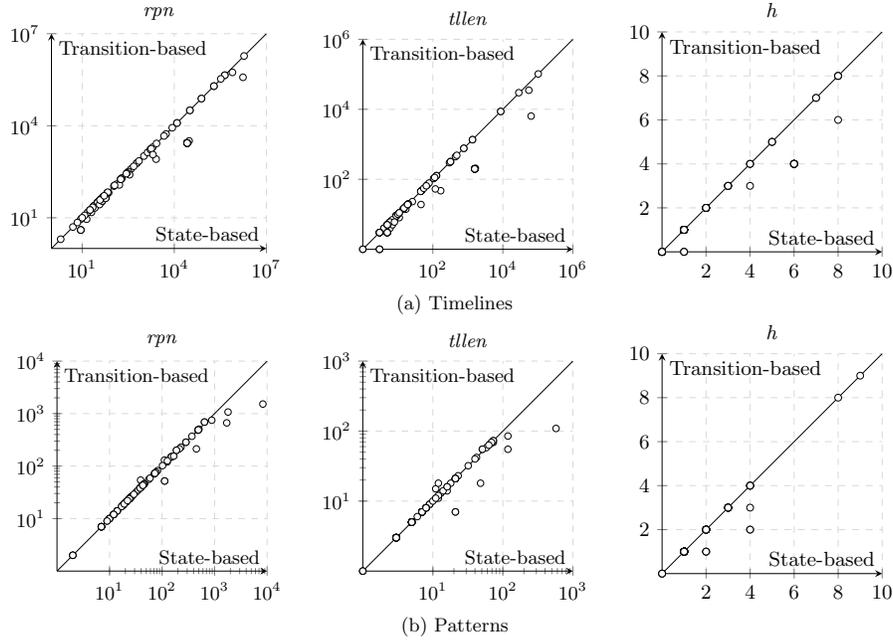

We also provide the results of additional experiments for \textbf{RQ2} by displaying the box plots for the patterns dataset in Fig.~\ref{fig:patterns-boxplots}.
% Begin boxplots of patterns dataset:
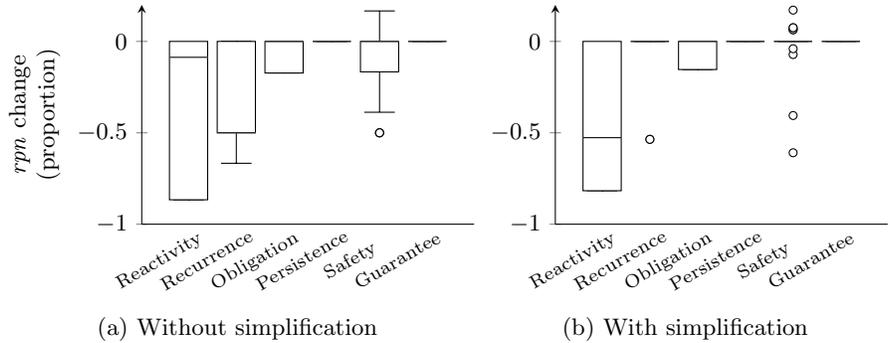
\begin{figure}[!ht]
    \centering
    % \captionsetup[subfigure]{labelformat=empty}
    % \resizebox{\textwidth}{!}{%
    \subfloat[Without simplification]{\begin{tikzpicture}
	\pgfplotstableread[col sep=comma]{data/patterns_mphierarchy.csv}\csvdata
    	\begin{axis}[
    		boxplot/draw direction = y,
    		% x axis line style = {opacity=0},
    		axis x line* = bottom,
    		axis y line = left,
    		% enlarge y limits,
    		% ymajorgrids,
    		xtick = {1, 2, 3, 4, 5, 6},
            xticklabel style={
                /pgf/number format/fixed,
                /pgf/number format/precision=1,
                /pgf/number format/fixed zerofill,
                % Adjust spacing as needed
                anchor=north,
                align=right,
                font=\scriptsize,
                yshift=-0.3cm,
                xshift=-0.5cm,
                rotate=30
            },
    		xticklabels = {Reactivity, Recurrence, Obligation, Persistence, Safety, Guarantee},
    		xtick style = {draw=none}, % Hide tick line
            ylabel style = {align=center},
    		ylabel = {\emph{rpn} change\\(proportion)},
            ymin=-1, ymax=0.2,
            width=6cm, height=4.5cm,
    	]
    		\foreach \n in {0,...,5} {
    			\addplot [-, boxplot, draw=black,mark=*, fill=white, mark size=1.5pt] table[y index=\n] {\csvdata};
    		}
    	\end{axis}
    \end{tikzpicture}\label{subfig:patterns_unsimpboxplot}} %
    \subfloat[With simplification]{    \begin{tikzpicture}
    	\pgfplotstableread[col sep=comma]{data/simppatterns_mphierarchy.csv}\csvdata
    	\begin{axis}[
    		boxplot/draw direction = y,
    		% x axis line style = {opacity=0},
    		axis x line* = bottom,
    		axis y line = left,
    		% enlarge y limits,
    		% ymajorgrids,
    		xtick = {1, 2, 3, 4, 5, 6},
            xticklabel style={
                /pgf/number format/fixed,
                /pgf/number format/precision=1,
                /pgf/number format/fixed zerofill,
                % Adjust spacing as needed
                anchor=north,
                align=right,
                font=\scriptsize,
                yshift=-0.3cm,
                xshift=-0.5cm,
                rotate=30
            },
    		xticklabels = {Reactivity, Recurrence, Obligation, Persistence, Safety, Guarantee},
    		xtick style = {draw=none}, % Hide tick line
            ymin=-1, ymax=0.2,
            width=6cm, height=4.5cm,
    	]
    		\foreach \n in {0,...,5} {
    			\addplot [-, boxplot, draw=black,mark=*, fill=white, mark size=1.5pt] table[y index=\n] {\csvdata};
    		}
    	\end{axis}
    \end{tikzpicture}
    \label{subfig:patterns_simpboxplot}
    }%}
    \caption{We observe a similar trend as the box plots in Fig.~\ref{fig:timelines-boxplots}. These plots present the proportional reductions in \emph{rpn} for the formulas in the patterns dataset grouped by their type.}
    \label{fig:patterns-boxplots}
\end{figure}

\clearpage

\bibliographystyle{splncs04}
\bibliography{refs}

\end{document}